\documentclass[runningheads]{llncs}
\usepackage{cite}
\usepackage{graphicx}
\usepackage{booktabs}   
\usepackage{subfig}
\usepackage{wrapfig}
\usepackage{multicol}
\usepackage{amsmath}
\usepackage{amsfonts}
\usepackage{tikz}
\usetikzlibrary{arrows.meta,automata,calc}
\usepackage{tikz-qtree} 
\usepackage{tikz-cd} 
\usepackage{algorithm}
\usepackage[noend]{algpseudocode}

\usepackage{syntax}

\usepackage{semantic}

\usepackage{marginnote}

\usepackage{listings}
\lstdefinelanguage{syncon}
  {morekeywords={syncon,infix,postfix,prefix,left,right,comment,precedence,forbid,except,type,builtin,rec,grouping,token},
   sensitive=true,
   morecomment=[l]{//},
   morecomment=[s]{/*}{*/},
   morecomment=[s][<.][.>],
   morestring=[b]",
  }
\newcommand{\ocaml}{\lstinline[language={[objective]caml}]}
\newcommand{\syncon}{\lstinline[language=syncon]}
\newcommand{\haskell}{\lstinline[language=haskell]}

\newcommand{\NT}{V} 
\newcommand{\T}{\Sigma} 
\newcommand{\Labels}{L} 
\newcommand{\yield}{\mathit{yield}} 
\newcommand{\semantic}{\mathit{unparen}} 
\newcommand{\parse}{\mathit{parse}} 
\newcommand{\words}{\mathit{words}} 
\newcommand{\amb}{\mathit{amb}}

\newcommand{\regex}{\mathit{Reg}}
\newcommand{\reqpl}{(}
\newcommand{\reqpr}{)}
\newcommand{\reqp}[1]{\reqpl#1\reqpr}
\newcommand{\pospl}{[}
\newcommand{\pospr}{]}
\newcommand{\posp}[1]{\pospl#1\pospr}
\newcommand{\dfa}{\mathit{dfa}} 
\newcommand{\labelnt}{\mathit{nt}} 
\newcommand{\labelregex}{\mathit{regex}} 
\newcommand{\treetorun}{\mathit{run}}
\newcommand{\runtotree}{\mathit{tree}}
\newcommand{\runtoword}{\mathit{word}}
\newcommand{\runprimetorun}{\mathit{normalrun}}
\newcommand{\shortest}{\mathit{shortest}}

\begin{document}

\title{Resolvable Ambiguity\thanks{This project is financially supported by the Swedish Foundation for Strategic Research (FFL15-0032)}}
%
%
\author{Viktor Palmkvist\inst{1} \and
Elias Castegren\inst{1} \and
Philipp Haller\inst{1} \and
David Broman\inst{1}}
%
\authorrunning{V. Palmkvist et al.}
%
\institute{KTH Royal Institute of Technology, 100 44 Stockholm, Sweden
\email{\{vipa,eliasca,dbro,phaller\}@kth.se}
}
\maketitle              
\begin{abstract}
A common standpoint when designing the syntax of programming languages is that
the grammar definition has to be unambiguous.  However, requiring up front unambiguous grammars
can force language designers to make more or less arbitrary choices to disambiguate
the language. In this paper, we depart from the traditional view of
unambiguous grammar design, and enable the detection of ambiguities to be delayed until
parse time, allowing the user of the language to perform the disambiguation. A
natural decision problem follows: given a language definition, can a
user always disambiguate an ambiguous program?  We introduce and
formalize this fundamental problem---called the \emph{resolvable
  ambiguity problem}---and divide it into separate static and dynamic
resolvability problems. We provide solutions to the static problem
for a restricted language class and sketch proofs of soundness and completeness. We also provide a sound
and complete solution to the dynamic problem for a much less
restricted class of languages.
The approach is evaluated through two separate case studies,
covering both a large existing programming language, and the
composability of domain-specific languages.
\keywords{Syntax \and Ambiguity \and Grammars \and Parsing \and Domain-specific languages}
\end{abstract}

\section{Introduction}

Ever since the early 60s, it has been known that determining whether a context-free
grammar is ambiguous is undecidable~\cite{cantorAmbiguityProblemBackus1962}. As a consequence,
a large number of restricted grammars have been developed to guarantee
that language definitions are unambiguous. This traditional view of
only allowing unambiguous grammars has been taken for granted as the
only truth: \emph{the} way of how the syntax of a language must be
defined ~\cite{sudkampLanguagesMachinesIntroduction1997,ahoCompilersPrinciplesTechniques2006,webberModernProgrammingLanguages2003,cooperEngineeringCompiler2011,ginsburgAmbiguityContextFree1966}.
However, our recent work~\cite{palmkvistCreatingDomainSpecificLanguages2019} suggests that
carefully designed ambiguous syntax definitions can be
preferable compared to completely unambiguous
definitions.

Another area where ambiguous grammars naturally arise is
in domain-specific language development.  An important
objective when designing domain-specific languages is to be able to construct
languages by composing and extending existing
languages.
Unfortunately, the composition of
%
grammars easily produces
an ambiguous grammar, even if the original grammars are all
unambiguous on their own. Approaches for solving this problem can be broadly split
into two categories: those that handle ambiguities on the
grammar-level (detection, prevention, etc.), and those that work on
particular examples of ambiguous programs. The former is well explored
in the form of heuristics for ambiguity detection
\cite{bastenAmbiguityDetectionProgramming2011,axelssonAnalyzingContextFreeGrammars2008,brabrandAnalyzingAmbiguityContextFree2007}
or restrictions on the composed grammars
\cite{kaminskiModularWellDefinednessAnalysis2013}, while the
latter has received relatively little attention.




This paper focuses on the problems that arise if a grammar definition
is not guaranteed to be unambiguous. Concretely, if parsing a program
could result in several correct parse trees, the grammar is obviously
ambiguous. The programmer then needs to disambiguate the program by,
for instance, inserting extra parentheses. We say that a program is
\emph{resolvably ambiguous} if the ambiguity can be resolved by
an end-user such that each of the possible parse trees can
be selected using different modifications.
We can divide the problem of resolvable ambiguity into two
different decision problems.
For a particular language, the \emph{dynamic resolvable ambiguity problem} asks the following question: can every parse tree of \emph{one particular} program be written unambiguously, while the \emph{static resolvability problem} asks if every parse tree of \emph{every} program can be written unambiguously. In particular, the latter problem is difficult, since it in general involves examining an infinite set (one per parse tree) of infinite sets (words that can parse as that parse tree) to find one without a unique element (i.e., an unambiguous word). Furthermore we must be able to produce an unambiguous program for a given parse tree to be able to give good error messages to an end-user encountering an ambiguity.

In this paper we describe a syntax definition formalism where ambiguities can be resolved by grouping parentheses, and precedence and associativity is specified through \emph{marks} (c.f. Section~\ref{sec:parse-time-disambiguation}). As a delimitation we only consider languages with balanced parentheses. Our approach builds on the following key insights: i) the language of a parse tree is visibly pushdown \cite{alurVisiblyPushdownLanguages2004}, ii) these languages can be encoded as words, and iii) the language of the word encodings of all trees for a language is also visibly pushdown. The last point allows the construction of a visibly pushdown automaton (VPDA) that we examine to solve the static problem for a particular language subclass. Our solution to the dynamic problem, which builds on the first point, is general.

A language designer can thus either use our solution to the static problem (if their language is in the appropriate subclass) or use an approach similar to unit testing through our solution to the dynamic problem. Our solution to the dynamic problem also gives suggested fixes to end-users encountering ambiguities.

More concretely, we make the following contributions:

\begin{itemize}
\item We formally define the term \emph{resolvable ambiguity}, and
  provide precise formalizations of the static and dynamic
  \emph{resolvable ambiguity problems}
  (Section~\ref{sec:resolvable-definition}).
\item As part of the formalization of the dynamic and static resolvability problems, we define a syntax definition formalism based on Extended Backus–Naur Form (EBNF) where ambiguities can be resolved with grouping parentheses (Section~\ref{sec:parse-time-disambiguation}).
\item We describe the language of a single, arbitrary parse tree and present a novel linear encoding of the same (Section~\ref{sec:linear-encoding}).
\item We devise a decidable algorithm that solves the \emph{dynamic resolvability problem} for languages with balanced parentheses (Section~\ref{sec:dynamic}). Additionally, this algorithm also produces a minimal unambiguous word for each resolvable parse tree.
\item We solve the \emph{static resolvability problem} for a language subclass. We further show that a resolvable result can be preserved through certain modifications that bring the language into a larger subclass (Section~\ref{sec:static}).
\item We evaluate the dynamic resolvability methodology in the context of two
  different case studies using a tool and a small DSL for defining
  syntax definitions (Section~\ref{sec:evaluation}): (i) a
  domain-specific language for orchestrating parallel
  computations, where we investigate the effects of composing
  language fragments, and (ii) an implementation and evaluation of
  a large subset of OCaml's syntax, where we study the effect of
  implementing an under-specified language definition.
\end{itemize}

\noindent Before presenting the above contributions, the paper starts with motivating examples (Section~\ref{sec:motivation}), as well as preliminaries (Section~\ref{sec:prel}). Finally, following a discussion of related work (Section~\ref{sec:related-work}), the paper provides some conclusions (Section~\ref{sec:conclusion}).

\section{Motivating Examples}
\label{sec:motivation}
In this section, we motivate the need for ambiguous language definitions, where the decision of how to disambiguate a program is taken by the end-user (the programmer) and not the language designer. We motivate our new methodology both for engineering of new domain-specific languages, as well as for the design of existing general-purpose programming languages.

\subsection{Domain-Specific Modeling of Components}
Suppose we are defining a new textual modeling language, where components are composed in series using an infix operator \verb~--~, and in parallel using an infix operator \verb~||~. For instance, an expression \verb~C1 -- C2~ puts components \verb~C1~ and \verb~C2~ in series, whereas \verb~C1 || C2~ composes them in parallel. In such a case, what is then the meaning of the following expression?
\begin{verbatim}
  C1 -- C2 || C3 || C4
\end{verbatim}

\noindent Are there natural associativity and precedence rules for these new operators? If there are no predefined rules of how to disambiguate this expression within the language definition, it is an ambiguous expression, and a parser generates a set of parse trees. Consider Figure~\ref{fig:circuits} which depicts four different alternatives, each with a different meaning, depending on how the ambiguity has been resolved. Clearly, the expression has totally different meanings depending on how the end-user places the parentheses. However, if a language designer is forced to make the grammar of the syntax definition unambiguous, a specific choice has to be made for precedence and associativity. For instance, assume that the designer makes the arbitrary choice that serial composition has higher precedence than parallel composition, and that both operators are left-associative. In such a case, the expression without parentheses is parsed as Figure~\ref{fig:circuits}(d). The question is why such an arbitrary choice---which is forced by the traditional design of unambiguous grammars---is the correct way to interpret a domain-specific expression. The alternative, as argued for in this paper, is to postpone the decision, and instead give an error message to the end-user (programmer or modeler), and expose different alternatives that disambiguate the expression.

\begin{figure*}[!t]
\center
\includegraphics[width=1.0\textwidth]{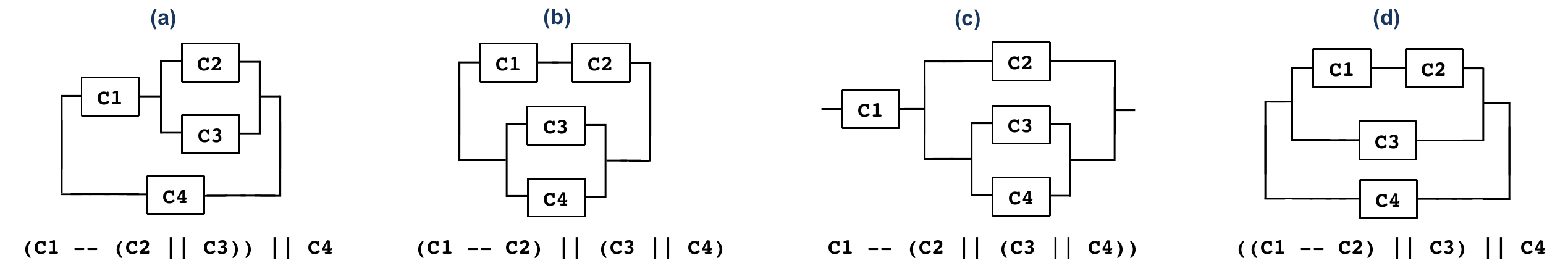}
\caption{The figure shows four different alternatives for disambiguating the expression $\texttt{C1 -- C2 || C3 || C4}$. Note that there is a fifth alternative $\texttt{C1 -- ((C2 || C3) || C4)}$. However, the meaning of this expression is the same as (c) assuming the parallel operator is associative. In such a case, this expression and the expression in (c) both mean that components $\texttt{C2}$, $\texttt{C3}$, and $\texttt{C4}$ are composed in parallel. }
\label{fig:circuits}
\end{figure*}


\subsection{Match Cases in OCaml}

Explicit ambiguity is highly relevant also for the design of new and existing general-purpose programming languages. The following example shows how an ambiguity error can be clearer than what an existing compiler produces today with the traditional approach.


Consider this OCaml example of nested \verb|match| expressions, as stated by \cite{palmkvistCreatingDomainSpecificLanguages2019}:

\begin{minipage}{.34\textwidth}
\begin{lstlisting}[language={[objective]caml},numbers=left,basicstyle=\small\tt]
match 1 with
  | 1 -> match "one" with
         | str -> str
  | 2 -> "two"
\end{lstlisting}
\end{minipage}
\vrule
\begin{minipage}{\textwidth}
\begin{lstlisting}[basicstyle=\small\tt]
 File "./nest.ml", line 4, characters 4-5:
 Error: This pattern matches values of type
        int but a pattern was expected which
        matches values of type string
\end{lstlisting}
\end{minipage}

\noindent The OCaml compiler output is listed to the right. The compiler sees the last line as belonging to the inner \ocaml{match} rather than the outer, as was intended. The solution is simple; we put parentheses around the inner match:

\begin{minipage}{\textwidth}
\begin{lstlisting}[language={[objective]caml},numbers=left,basicstyle=\small\tt]
match 1 with
  | 1 -> (match "one" with
          | str -> str)
  | 2 -> "two"
\end{lstlisting}
\end{minipage}

\noindent However, the connection between the error message and the solution is not particularly clear; surrounding an expression with parentheses does not change its type.
%
%
%
%
%
%
The OCaml compiler makes an arbitrary choice to remove the ambiguity, which may or may not be the alternative the user intended. With a parser that is aware of possible ambiguities, the disambiguation can be left to the end-user, with the alternatives listed as part of an error message.



\subsection{Language Composition}

The possibility of composing different languages is a prevalent idea in research \cite{vanwykSilverExtensibleAttribute2010,lorenzenSoundTypedependentSyntactic2016,danielssonParsingMixfixOperators2011}. Particularly relevant to this paper is the behavior of composed \emph{syntax}; composing unambiguous syntaxes does not necessarily produce an unambiguous composed syntax. Pre-existing systems commonly solve this either by strict requirements on the individual languages \cite{vanwykSilverExtensibleAttribute2010} or by largely ignoring the problem and merely presenting the multiple syntax trees in case of an ambiguity \cite{lorenzenSoundTypedependentSyntactic2016,danielssonParsingMixfixOperators2011}. The former case gives good guarantees for the end-user experience (no ambiguities) while the latter gives more freedom to the language designer. Resolvable ambiguity, as presented in this paper, provides a middle ground with more language designer freedom than the former and a different guarantee for the end-user: no \emph{unresolvable} ambiguities.

Language composition has limited prevalence outside of research, but it does exist in a simplified form: custom operators in libraries. For example, Haskell allows a programmer to define new operators, with new syntax and custom precedence level and associativity. A user can then import these operators from any number of libraries and use them together in one program. Precedence in Haskell is defined as a number between 0 and 9, meaning that operators from different libraries have a defined relative precedence. If the libraries are designed by different authors then this precedence is unlikely to have been considered and may or may not be sensible. Leaving the relative precedence undefined would not affect expressivity, every possible combination of operators can be written by simply adding parentheses, and would avoid surprising interpretations. This does in fact happen in some cases in Haskell, namely when two operators have the same precedence but incompatible associativity. Both \haskell{>>=} and \haskell{=<<} have precedence 1, but the former is left-associative while the latter is right-associative. \haskell{a >>= b =<< c} then produces the following error:

\begin{lstlisting}
Precedence parsing error
  cannot mix '>>=' [infixl 1] and '=<<' [infixr 1]
  in the same infix expression
\end{lstlisting}

\noindent The error message shows the problem but suggests no solution. Our approach can suggest solutions, allow arbitrary operators to have undefined relative precedence (intransitive precedence), and is general; it handles non-operator ambiguities as well.

\section{Notation}
\label{sec:prel}

This section briefly explains the notations used in this paper.
Appendix~\ref{app:prel} contains a more detailed description of
the preliminaries.

\paragraph{Grammars}
Given a context-free grammar $G$, we write $L(G)$ to denote the
set of all words recognized by $G$.
The standard definition of ambiguity states that a word
$w \in L(G)$ is ambiguous if there are two distinct leftmost
derivations of that word.

\paragraph{Automata}
An automaton is a 5-tuple $(Q,\T,\delta,q_0,F)$, where $Q$ is a
finite set of states; $\T$ a finite set of terminals; $\delta$ a
transition function from $Q \times \T$ to finite subsets of $Q$
(or single states of $Q$ for a deterministic automaton);
$q_0 \in Q$ an initial state; and $F \subseteq Q$ a set of final
states.
A pushdown automata further adds a set of stack symbols $\Gamma$
and equips the transition function $\delta$ with information on
which symbols are pushed, popped or just read from the stack.
%

\paragraph{Unranked Regular Tree Grammars}
Trees generalize words by allowing each terminal to have multiple
ordered successors, instead of just zero or one. In this paper, we
are only concerned with \emph{unranked} trees, where the arity of
a terminal is not fixed.
The allowed sequence of ordered successors is described by a regular language~\cite{comonTreeAutomataTechniques2007}. For
example, $a(n(\verb|'1'|) \;\verb|'+'|\; n(\verb|'2'|))$ represents
the parse tree of the string ``$1+2$''.
Finally, $yield : L(T) -> \T^{*}$ is the sequence of terminals
obtained by a left-to-right traversal of the leaves of a tree. Informally, it is
the flattening of a tree after all internal nodes have been
removed.

\section{Resolvable Ambiguity} \label{sec:resolvable-definition}

This section introduces our definition of \emph{resolvable ambiguity}, and then relates it to standard concepts in formal languages.

A formal language is defined as a set of words, i.e., a subset of $\T^{*}$ for some alphabet $\T$. To be able to define \emph{resolvable ambiguity}, we additionally have to consider the results of parsing words. In order to stay as general as possible, we first define the notion of an \emph{abstract parser}.

\begin{definition}
  An abstract parser $P$ is a triple $(\T, T, \parse)$ consisting of
\begin{itemize}
\item an alphabet $\T$,
\item a set of parse trees $T$, and
\item a function $\parse : \T^{*} -> 2^T$ that relates words to their parse trees, where $2^T$ denotes the powerset of $T$. Additionally, we require that $T = \bigcup_{w \in \T^{*}} \parse(w)$.
\end{itemize}
\end{definition}

\noindent Note that we do not require each tree to have a unique word, i.e., there may exist two distinct $w_1$ and $w_2$ such that $t \in \parse(w_1)$ and $t \in \parse(w_2)$. This notion of an abstract parser enables the introduction of a particular class of formal languages that we will use throughout the rest of the paper; we call members of this class \emph{parse languages:}

\begin{definition}
  Given a word $w \in \T^{*}$ and an abstract parser $P = (\T, T,
  \parse)$, the set of words contained in the \emph{parse language}
  $L(P)$ is defined as follows:

  $w \in L(P)$ iff $\parse(w) \neq \emptyset$.
\end{definition}

\noindent For example, consider a simple arithmetic language without precedence and parentheses. In such a language, $\parse(1 + 2 \cdot 3)$ would produce a set containing two parse trees. 
%
%
The ambiguity of a word $w \in \T^{*}$ is defined in terms of $\parse$:

\begin{definition}
  Given an abstract parser $P = (\T, T, \parse)$, a word $w \in L(P)$
  is ambiguous, written $\amb_P(w)$, iff
  $\exists t_1, t_2 \in T.\ t_1 \neq t_2 \land \{t_1, t_2\} \subseteq \parse(w)$
\end{definition}

\noindent
Note that the above definition implies that a word $w \in L(P)$, where
$P = (\T, T, \parse)$, is not ambiguous, or \emph{unambiguous}, if
$\exists t \in T.\ \parse(w) = \{t\}$.

We can connect the above definition of a parse language to the
classical definition of a (formal) language as follows:

\begin{itemize}
\item Given a parse language $L(P)$, the corresponding classical language (i.e., set of words) is given by $\{ w \mid \parse(w) \neq \emptyset \}$.
\item If we select a $\parse$ function that relates words to their leftmost derivations in a given context-free grammar, then our definition of ambiguity corresponds exactly to the classical definition of ambiguity.
\end{itemize}

\noindent For \emph{resolvable ambiguity} the definition instead centers around parse trees:

\begin{definition}\label{def:resolvable-tree}
  Given an abstract parser $P = (\T, T, \parse)$, a tree $t \in T$ is \emph{resolvably ambiguous}, written $\rho_P(t)$, iff

  $\exists w \in \T^{*}.\ \parse(w') = \{t\}$.
\end{definition}

\noindent A word is then resolvably ambiguous if all its parse trees are resolvably ambiguous:

\begin{definition}\label{def:resolvable-word}
  Given an abstract parser $P = (\T, T, \parse)$, a word $w \in L(P)$ is \emph{resolvably ambiguous}, written $\rho_P(w)$, iff

  $\forall t \in \parse(w).\ \rho_P(t)$.
\end{definition}

\noindent Additionally, we define an \emph{abstract parser} to be resolvably ambiguous if all its parse trees are resolvably ambiguous, formally:

\begin{definition}\label{def:resolvable-language}
  An abstract parser $P$ is resolvably ambiguous, written $\rho(P)$, iff

  $\forall t \in T.\ \rho_P(t)$.
\end{definition}

\noindent We can now make the following observations:

\begin{itemize}
\item An unambigous word $w$ is trivially resolvably ambiguous, since its only parse tree $t$ can be written unambiguously with $w$ itself ($\parse(w) = \{t\}$). The set of resolvably ambiguous words is thus a superset of the unambiguous words.
\item If a given parse tree $t$ has only one word $w$ such that $t \in \parse(w)$, then $w$ is resolvably ambiguous iff it is unambiguous. In general, $\forall t \in T.\ \lvert\{w \mid t \in \parse(w)\}\rvert = 1$ implies that the set of resolvably ambiguous words is exactly the set of unambiguous words.
\end{itemize}

\noindent The second point suggests that resolvable ambiguity is only an interesting property if an element of $T$ does not uniquely identify an element of $\T^{*}$. Intuitively, this only happens if $\parse$ discards some information present in its argument when constructing an individual parse tree. Fortunately, this is generally true for parsing in commonly used programming languages; they tend to discard, e.g., grouping parentheses and whitespace. In general, whatever information $\parse$ discards can be used by an end-user to disambiguate an ambiguous program.

We thus propose to loosen the common ``no ambiguity'' restriction on programming language grammars, and instead only require them to be resolvably ambiguous. However, merely having an arbitrary function $\parse$ gives us very little to work with, and no way to decide whether the language it defines is resolvably ambiguous or not. The remainder of this paper will thus consider $\parse$ functions defined using a particular formalism, introduced in Section~\ref{sec:parse-time-disambiguation}, which gives us some decidable properties.

Before introducing this formalism, however, we introduce the two central problems we consider in this paper:

\begin{description}
\item[Static resolvability.] Given an abstract parser $P$, determine whether $\rho(P)$.
\item[Dynamic resolvability.] Given an abstract parser $P$ and a word $w \in L(P)$, determine whether $\rho_P(w)$.
\end{description}

\noindent Our main concern is producing decision procedures for these problems. First, we define soundness and completeness for the static problem.

\begin{definition}[Soundness of Static Resolvability]\label{def:static-procedure-sound}
  A decision procedure $f$, solving the static resolvability problem, is sound iff

  $f(P) \implies \rho(P) \text{ for all abstract parsers } P$
\end{definition}

\begin{definition}[Completeness of Static Resolvability]\label{def:static-procedure-complete}
  A decision procedure $f$, solving the static resolvability problem, is \emph{complete} iff

  $\rho(P) \implies f(P) \text{ for all abstract parsers } P$
\end{definition}

\noindent Similarly, we define soundness and completeness for the dynamic problem.

\begin{definition}[Soundness of Dynamic Resolvability]\label{def:dynamic-procedure-sound}
  A decision procedure $f$, solving the dynamic resolvability problem, is sound iff

  $f(P, w) \implies \rho_P(w) \text{ for all } w \in \Sigma^{*} \text{ and abstract parsers } P = (\T, T, \parse)$
\end{definition}

\begin{definition}[Completeness of Dynamic Resolvability]\label{def:dynamic-procedure-complete}
  A decision procedure $f$, solving the dynamic resolvability problem, is complete iff

  $\rho_P(w) \implies f(P, w) \text{ for all } w \in \Sigma^{*} \text{ and abstract parsers } P = (\T, T, \parse)$
\end{definition}

\section{Parse-time Disambiguation} \label{sec:parse-time-disambiguation}

This section describes our chosen language definition formalism, and motivates its design.

The primary purpose of this formalism is, as described in the previous section, to produce a $\parse$ function, i.e., to describe a word language and assign one or more parse trees to each word. Furthermore, Section~\ref{sec:resolvable-definition} suggests that disambiguation is made possible by letting $\parse$ discard information. We use unranked trees as parse trees and have $\parse$ discard grouping parentheses.

With that in mind, we define a language definition $D$ as a set of labelled productions, as described in Fig.~\ref{fig:input-language-definition}. Note that we require the labels to uniquely identify the production, i.e., there can be no two distinct productions in $D$ that share the same label. Also note that the right-hand side of a production is a regular expression, rather than the theoretically simpler sequence used in a context-free grammar.

Each non-terminal appearing in the regular expression of a production carries a \emph{mark} $m$ which is a set of labels whose productions may \emph{not} replace that non-terminal. To lessen clutter, we write $E_\emptyset$ as $E$. Consider the language definition shown in Fig.~\ref{fig:running-example-definition} which we use as a running example. In the production describing multiplication ($m$) both non-terminals are marked with $\{a\}$, which thus forbids addition from being a direct child of a multiplication. By ``direct child'' we mean ``without an intermediate node'', most commonly grouping parentheses; thus, this enforces conventional precedence.

\begin{figure}[b]

\centering
\begin{minipage}{.47\textwidth}
  \centering
  \begin{tabular}{@{}ll@{}}
      Terminals & $t \in \T$ \\
      Non-terminals & $N \in \NT$ \\
      Labels & $l \in \Labels$ \\
      \multicolumn{2}{@{}l@{}}{$\T$, $\NT$, and $\Labels$ disjoint} \\
      \addlinespace
      Marks & $m \subseteq \Labels$ \\
      Regular expressions & $r ::= t \mid N_m \mid r \cdot r $ \\
      & \hphantom{$r ::=$}$\mid r + r \mid \epsilon \mid r^{*}$ \\
      Labelled productions & $N -> l : r$ \\
  \end{tabular}
  \captionof{figure}{The abstract syntax of a language definition.}
  \label{fig:input-language-definition}
\end{minipage}\quad
\begin{minipage}{.45\textwidth}
  \centering
  \begin{tabular}{@{}l@{\quad$->$\quad}c@{ $:$\quad}l@{}}
    $E$ & $l$ & \verb|'['| ($E$ (\verb|';'| $E$)$^{*}$ $+$ $\epsilon$) \verb|']'| \\
    $E$ & $a$ & $E$ \verb|'+'| $E$ \\
    $E$ & $m$ & $E_{\{a\}}$ \verb|'*'| $E_{\{a\}}$ \\
    $E$ & $n$ & $I$ \\
  \end{tabular}
  \captionof{figure}{The input language definition used as a running example, an expression language with lists, addition, and multiplication, with precedence defined, but not associativity. Assumes that $I$ matches a numeric terminal.}
  \label{fig:running-example-definition}
\end{minipage}
\end{figure}

From $D$ we then generate four grammars: $G_D$, $T_D$, $G'_D$, and $T'_D$. Technically, only $G'_D$ and $T_D$ are required, $G'_D$ is used as the defined word language and $T_D$ as the parse trees, but the remaining two grammars help the presentation.

\begin{itemize}
\item $G_D$ represents a word language describing all semantically distinct programs.
\item $T_D$ represents a tree language describing the parse trees of words in $L(G_D)$.
\item $G'_D$ is essentially a modified version of $G_D$, e.g., adding grouping parentheses and other forms of disambiguation (i.e., the result of marks).
\item $T'_D$ represents a tree language describing the parse trees of words in $L(G'_D)$.
\end{itemize}

\noindent Fig.~\ref{fig:running-example-generated} shows the four grammars generated from our running example in Fig.~\ref{fig:running-example-definition}. The context-free grammars are produced by a rather standard translation from regular expressions to CFGs, while the primed grammars get a new non-terminal per distinctly marked non-terminal in $D$, where each new non-terminal only has the productions whose label is not in the mark. For example, the non-terminal $E_{\{a\}}$ in Fig.~\ref{fig:running-example-generated:t-prime} has no production corresponding to the $a$ production in Fig.~\ref{fig:running-example-definition}.

\begin{figure}[t]
  \subfloat[$T_D$, the parse trees of $G_D$.\label{fig:running-example-generated:t}]{%
    \begin{minipage}{0.5\textwidth}
    \centering
    \begin{tabular}{@{}l@{\quad$->$\quad}c@{$($ }l@{}}
      \toprule
      $E$ & $l$ & \texttt{'['} $(\epsilon + E($\texttt{';'} $E)^{*})$ \texttt{']'} $)$ \\
      $E$ & $a$ & $E$ \texttt{'+'} $E$ $)$ \\
      $E$ & $m$ & $E$ \texttt{'*'} $E$ $)$ \\
      $E$ & $n$ & $I$ $)$ \\
      \bottomrule
    \end{tabular}
    \end{minipage}
  }
  \subfloat[$T'_D$, the parse trees of $G'_D$.\label{fig:running-example-generated:t-prime}]{%
    \begin{minipage}{0.5\textwidth}
    \centering
    \begin{tabular}{@{}l@{\quad$->$\quad}c@{$($ }l@{}}
      \toprule
      $E$ & $l$ & \texttt{'['} $(\epsilon + E($\texttt{';'} $E)^{*})$ \texttt{']'} $)$ \\
      $E$ & $a$ & $E$ \texttt{'+'} $E$ $)$ \\
      $E$ & $m$ & $E_{\{a\}}$ \texttt{'*'} $E_{\{a\}}$ $)$ \\
      $E$ & $n$ & $I$ $)$ \\
      $E$ & $g$ & \texttt{'('} $E$ \texttt{')'} $)$ \\
      \midrule
      $E_{\{a\}}$ & $l$ & \texttt{'['} $(\epsilon + E($\texttt{';'} $E)^{*})$ \texttt{']'} $)$ \\
      $E_{\{a\}}$ & $m$ & $E_{\{a\}}$ \texttt{'*'} $E_{\{a\}}$ $)$ \\
      $E_{\{a\}}$ & $n$ & $I$ $)$ \\
      $E_{\{a\}}$ & $g$ & \texttt{'('} $E$ \texttt{')'} $)$ \\
      \bottomrule
    \end{tabular}
    \end{minipage}
  }

  \subfloat[$G_D$, the generated abstract syntax.\label{fig:running-example-generated:w}]{%
    \begin{minipage}{0.5\textwidth}
    \centering
    \begin{tabular}{@{}l@{\quad$->$\quad}l@{}}
      \toprule
      $E$ & \texttt{'['} $E_{l1}$ \texttt{']'} \\
      $E$ & $E$ \texttt{'+'} $E$ \\
      $E$ & $E$ \texttt{'*'} $E$ \\
      $E$ & $I$ \\
      \midrule
      $E_{l1}$ & $\epsilon$ \\
      $E_{l1}$ & $E$ $E_{l2}$ \\
      \midrule
      $E_{l2}$ & $\epsilon$ \\
      $E_{l2}$ & \texttt{';'} $E$ $E_{l2}$ \\
      \bottomrule
    \end{tabular}
    \end{minipage}
  }
  \subfloat[$G'_D$, the generated concrete syntax.\label{fig:running-example-generated:w-prime}]{%
    \begin{minipage}{0.5\textwidth}
    \centering
    \begin{tabular}{@{}l@{\quad$->$\quad}l@{}}
      \toprule
      $E$ & \texttt{'['} $E_{l1}$ \texttt{']'} \\
      $E$ & $E$ \texttt{'+'} $E$ \\
      $E$ & $E_{\{a\}}$ \texttt{'*'} $E_{\{a\}}$ \\
      $E$ & $I$ \\
      $E$ & \texttt{'('} $E$ \texttt{')'} \\
      \midrule
      $E_{\{a\}}$ & \texttt{'['} $E_{l1}$ \texttt{']'} \\
      $E_{\{a\}}$ & $E_{\{a\}}$ \texttt{'*'} $E_{\{a\}}$ \\
      $E_{\{a\}}$ & $I$ \\
      $E_{\{a\}}$ & \texttt{'('} $E$ \texttt{')'} \\
      \midrule
      $E_{l1}$ & $\epsilon$ \\
      $E_{l1}$ & $E$ $E_{l2}$ \\
      \midrule
      $E_{l2}$ & $\epsilon$ \\
      $E_{l2}$ & \texttt{';'} $E$ $E_{l2}$ \\
      \bottomrule
    \end{tabular}
    \end{minipage}
  }
  \caption{The generated grammars.}
  \label{fig:running-example-generated}
\end{figure}

Examples of elements in each of these four languages can be seen in Fig.~\ref{fig:example-words-in-square}. Each element corresponds to the word ``$(1 + 2) * 3$'' in $L(G'_D)$. Note that the word in $L(G_D)$ is ambiguous, and that there are other words in $L(G'_D)$ that correspond to the same element in $L(T_D)$, e.g., ``$((1 + 2)) * 3$'' and ``$(1 + 2) * (3)$''. As a memory aid, the primed versions ($G'_D$ and $T'_D$) contain disambiguation (grouping parentheses, precedence, associativity, etc.) while the unprimed versions ($G_D$ and $T_D$) are the (likely ambiguous) straightforward translations (i.e., ignoring marks) from $D$.

{
\newcommand{\terminal}[1]{\ \underline{#1}\ }

\begin{figure}
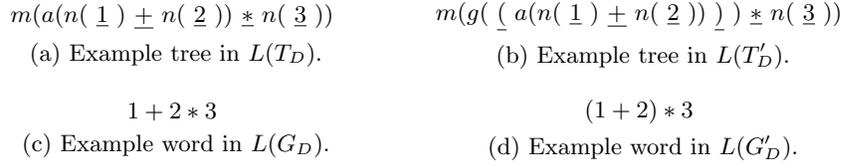

  \subfloat[Example tree in $L(T_D)$.\label{fig:sum-add-tree}]{
    \begin{minipage}{.5\linewidth}
      \centering
      $m(a(n(\terminal{1}) \terminal{+} n(\terminal{2})) \terminal{*} n(\terminal{3}))$
    \end{minipage}
  }
  \subfloat[Example tree in $L(T'_D)$.]{
    \begin{minipage}{.5\linewidth}
      \centering
      $m(g( \terminal{(} a(n( \terminal{1} ) \terminal{+} n( \terminal{2} )) \terminal{)} ) \terminal{*} n( \terminal{3} ))$
    \end{minipage}
  }

  \subfloat[Example word in $L(G_D)$.]{%
    \begin{minipage}{.5\linewidth}
      \centering
      $1 + 2 * 3$
    \end{minipage}
  }
  \subfloat[Example word in $L(G'_D)$.]{%
    \begin{minipage}{.5\linewidth}
      \centering
      $(1 + 2) * 3$
    \end{minipage}
  }
  \caption{Example with elements from each generated language that correspond to each other. The leaf terminals in the tree languages appear underlined to distinguish the two kinds of parentheses.}
  \label{fig:example-words-in-square}
\end{figure}
}

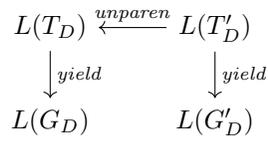
\begin{wrapfigure}{l}{0.3\textwidth}
  \begin{tikzcd}
    L(T_D) \arrow[d, "\yield"] & L(T'_D) \arrow[l, "\semantic"'] \arrow[d, "\yield"] \\
    L(G_D) & L(G'_D)
  \end{tikzcd}
  \caption{The generated grammars, and their relation to each other.}
  \label{fig:grammar-square}
\end{wrapfigure}

At this point we also note that the shape of $D$ determines where the final concrete syntax permits grouping parentheses; they are allowed exactly where they would surround a complete production. For example, $G_D$ in Fig.~\ref{fig:running-example-generated:w} can be seen as a valid language definition (if we generate new unique labels for each of the productions). However, starting with that language definition would allow the expression ``$[1(;2)]$'', which makes no intuitive sense; grouping parentheses should only be allowed around complete expressions, but ``$;2$'' is not a valid expression.

Finally, we require a function $\semantic : L(T'_D) -> L(T_D)$ that removes grouping parentheses from a parse tree, i.e., it replaces every subtree $g($ \verb|'('| $t$ \verb|')'| $)$ with $t$. The relation between the four grammars in terms of $\yield$ and $\semantic$ can be seen in Fig.~\ref{fig:grammar-square}. With this we can define $\parse : L(G'_D) -> 2^{L(T_D)}$, along with its inverse $\words : L(T_D) -> 2^{L(G'_D)}$:

$$
\begin{array}{rcl}
\parse(w) & = & \{ \semantic(t) \mid t \in L(T'_D) \land \yield(t) = w \} \\
\words(t) & = & \{ w \mid t \in \parse(w) \} \\
\end{array}
$$

\noindent The latter is mostly useful in later sections, but $\parse$ allows us to consider some concrete examples of resolvable and unresolvable ambiguities. For example, in our running example (Fig.~\ref{fig:running-example-definition}), the word \verb|'1 + 2 + 3'| is ambiguous, since $\parse($\verb|'1 + 2 + 3'|$) = \{t_1, t_2\}$ where

\begin{center}
  \begin{tabular}{l}
    $t_1 = a($ \hphantom{$a($} $n($ \verb|'1'| $)$ \verb|'+'| $a($ $n($ \verb|'2'| $)$ \hphantom{$)$} \verb|'+'| $n($ \verb|'3'| $)$ $)$ $)$ \\
    $t_2 = a($ $a($ $n($ \verb|'1'| $)$ \verb|'+'| \hphantom{$a($} $n($ \verb|'2'| $)$ $)$ \verb|'+'| $n($ \verb|'3'| $)$ \hphantom{$)$} $)$ \\
  \end{tabular}
\end{center}

\noindent This is a resolvable ambiguity, since $\parse($\verb|'1 + (2 + 3)'|$) = \{t_1\}$ and \\$\parse($\verb|'(1 + 2) + 3'|$) = \{t_2\}$. To demonstrate an unresolvable case, we add the production $E -> s: E$ \verb|';'| $E$ (common syntax for sequential composition), at which point we find that the word \verb|'[1 ; 2]'| is unresolvably ambiguous; $\parse($\verb|'[1 ; 2]'|$) = \{t_3, t_4\}$ where:

\begin{center}
  \begin{tabular}{l}
    $t_3 = l($ \verb|'['| \hphantom{$s($} $n($ \verb|'1'| $)$ \verb|';'| $n($ \verb|'2'| $)$ \hphantom{$)$} \verb|']'| $)$ \\
    $t_4 = l($ \verb|'['| $s($ $n($ \verb|'1'| $)$ \verb|';'| $n($ \verb|'2'| $)$ $)$ \verb|']'| $)$ \\
  \end{tabular}
\end{center}

\noindent In this case, $t_4$ has an unambigous word (namely \verb|'[(1 ; 2)]'|), but $t_3$ does not. The solution is to forbid list elements from being sequences by modifying the language definition in Fig.~\ref{fig:running-example-definition} so that both non-terminals in the production $l$ are marked with $s$ (i.e., they look like $E_{\{s\}}$), at which point $\parse($\verb|'[1 ; 2]'|$) = \{t_3\}$.

\subsection{The Word Language of a Parse Tree} \label{sec:linear-encoding}

Central to our approach is the shape of $\words(t)$. Consider the tree corresponding to \verb|'(1 + 2) * 3'|. Each node in the tree may be surrounded by zero or more parentheses, except \verb|'1 + 2'|, which requires at least one pair. Its language is thus represented by $\{(^{n_1}(^{n_2}(^{n_3}1)^{n_3} + (^{n_4}2)^{n_4})^{n_2} * (^{n_5}2)^{n_5})^{n_1} \mid n_i \in \{0, 1, \ldots\}, n_2 \geq 1\}$. To have something more manageable we introduce an alternative representation: a linear encoding as a word. This representation is convenient for Section~\ref{sec:dynamic} and essential for Section~\ref{sec:static}.

Continuing with the example above, we can encode this language as a word, if we take ``$\reqp{}$'' to mean ``exactly one pair of parentheses'' and ``$\posp{}$'' to mean ``zero or more pairs of parentheses'': ``$\posp{\posp{\reqp{\posp{1} + \posp{2}}} * \posp{3}}$''. This encoding lets us reduce comparisons between languages to comparisons between words, with one caveat: the encoding is not unique. For example, ``$\reqp{\posp{}}$'' encodes the same language as ``$\posp{\reqp{}}$'', and ``$\posp{\posp{}}$'' encodes the same language as ``$\posp{}$''. We rectify this by repeatedly swapping ``$\reqp{\posp{\ldots}}$'' with ``$\posp{\reqp{\ldots}}$'' and ``$\posp{\posp{\ldots}}$'' with ``$\posp{\ldots}$'' until we reach a fixpoint. We call the result the \emph{canonical linear encoding} of the (word) language of a tree in $L(T_D)$, which \emph{is} unique.

We are now ready to construct decision procedures for the static and dynamic resolvability problems, as given in Definitions~\ref{def:static-procedure-sound}, \ref{def:static-procedure-complete}, \ref{def:dynamic-procedure-sound}, and \ref{def:dynamic-procedure-complete}. Section~\ref{sec:static} solves the static problem for a language subclass, while Section~\ref{sec:dynamic} fully solves the dynamic problem.

\section{Dynamic Resolvability Analysis} \label{sec:dynamic}

The dynamic resolvability problem is as follows: for a given $w' \in L(G'_D)$ determine whether $\forall t \in \parse(w').\ \exists w'_2.\ \parse(w'_2) = \{t\}$. Furthermore, for practical reasons, if the word is resolvably ambiguous we wish to produce a (minimal) witness for each tree. We place two restriction on languages $D$ we consider:

\begin{enumerate}
\item Stemming from our initial delimitation, each right-hand side regular expression must only recognize words with balanced parentheses. For example, ``$\verb|(| \verb|)|$'' is permissible, but ``$\verb|(|^{*} \verb|)|^{*}$'' is not.
\item $G_D$, but with parentheses removed, must not be infinitely ambiguous.
\end{enumerate}

\noindent Both of these restrictions can be checked statically. Additionally, if we already know statically that $D$ is resolvably ambiguous\footnote{While the \emph{decision} part of this procedure is uninteresting in this case (it will always answer ``resolvable''), the minimal witness \emph{is} interesting: it shows the user how to resolve the ambiguity.} we can drop the second requirement.

Our approach centers around the construction of a VPDA recognizing $\words(t)$ for any particular $t \in L(T_D)$. We can easily construct this automaton via the canonical linear encoding: let every pair ``$\reqp{}$'' push and pop the same stack symbol (call it $\gamma$) and let every ``$\posp{}$'' push and pop a unique stack symbol. For example, ``$\posp{\reqp{a}\posp{b}}$'' produces the following automaton:

\begin{center}
\begin{tikzpicture}[->,>=Latex,shorten >=1pt,auto,node distance=2cm,
    scale = .9,transform shape]
  \node[state,initial] (1) {};
  \node[state] (2) [right of=1] {};
  \node[state] (3) [right of=2] {};
  \node[state] (4) [right of=3] {};
  \node[state,accepting] (5) [right of=4] {};

  \path (1) edge[loop above] node{$'(', +1$} (1)
  (1) edge node{$'(', +\gamma$} (2)
  (2) edge node{$'a'$} (3)
  (3) edge node{$')', -\gamma$} (4)
  (4) edge[loop above] node{$'(', +2$} (4)
  (4) edge node{$'b'$} (5)
  (5) edge[loop above] node{$')', -2$} (5)
  (5) edge[loop right] node{$')', -1$} (5);
\end{tikzpicture}
\end{center}

\noindent This is useful to us since VPDAs are closed under difference. Given a pair of trees $t_1$ and $t_2$ we can thus produce a pair of automata recognizing $\words(t_1)$ and $\words(t_2)$ respectively, then construct a pair of automata recognizing $\words(t_1) \setminus \words(t_2)$ and $\words(t_2) \setminus \words(t_1)$, respectively. These new automata recognize the words that are not ambiguous between these two particular trees. Our algorithm can thus be seen as Algorithm~\ref{alg:dynamic}.

\begin{algorithm}
\caption{The Dynamic Algorithm}
\label{alg:dynamic}
\begin{algorithmic}[1]
\Procedure{Dynamic}{$T$}
  \State $R <- \emptyset$
  \State $U <- \emptyset$
  \ForAll{$t \in T$}
    \State $T' <- T \setminus \{t\}$
    \State $A <- \words(t)$ \Comment{$A$ is a VPDA.}
    \Loop
      \State $A <- A \setminus \bigcup_{t' \in T'} \words(t')$ \Comment{VPDAs closed under difference and union.}
      \If{$L(A) = \emptyset$}
        \State $U <- U \cup \{t\}$
        \State $\textbf{break loop}$
      \EndIf
      \If{$\parse(\shortest(A)) = \{t\}$} \Comment{$\shortest(A)$ denotes a shortest}
        \State $R <- R \cup \{(t, \shortest(A))\}$ \Comment{word recognized by $A$.}
        \State $\textbf{break loop}$
      \EndIf
      \State $T' <- \parse(\shortest(A)) \setminus \{t\}$
    \EndLoop
  \EndFor
  \State $\textbf{return} (R, U)$
\EndProcedure
\end{algorithmic}
\end{algorithm}

Note that the shortest word we produce might not be unique, in the general case there may be more than one shortest word that disambiguates a tree.

We require two things to ensure termination: that $T'$ is finite, and that the inner loop terminates. $T'$ is finite if no words are infinitely ambiguous, which is ensured by requirement 2. It is also ensured if $D$ is resolvably ambiguous, since an infinite ambiguity requires that $G_D$ contains a cycle $N =>^{+} N$, which would imply that any parse tree containing a production from $N$ only has (infinitely) ambiguous words, i.e., contradiction. Requirement 2 also ensures that the number of trees that share an ambiguous word with $t$ is finite (since each tree corresponds to a left-most derivation in $G_D$ and must share exactly the same non-parenthesis terminals).

\section{Static Resolvability Analysis} \label{sec:static}

To determine if a given language definition $D$ is resolvably ambiguous we attempt to find a counterexample: a tree $t \in L(T_D)$ such that there is no $w \in L(G'_D)$ for which $\parse(w) = \{t\}$, or prove that no such tree exists. Or, more briefly put: find a tree that has only ambiguous words or show that no such tree exists.

Our approach finds this counterexample by finding a pair of trees $t_1$ and $t_2$ such that $\words(t_1) \subseteq \words(t_2)$. Such a pair trivially implies that $t_1$ has no unambiguous words, since $\forall w.\ t_1 \in \parse(w) -> t_2 \in \parse(w)$, which thus implies that the examined language is unresolvably ambiguous. The converse is less obvious, in fact, the absence of such a pair does not guarantee a resolvably ambiguous language in general. We shall however show it to be sufficient for a limited class of languages.

We now divide the possible languages along two axes: whether a language has (non-grouping) parentheses, and whether a language has marks. The parentheses present in the canonical encoding will have the following shapes ($\mathbb{N} = \{0,1,\ldots\}$):

\begin{center}
\begin{tabular}{r@{\quad}c@{\quad}c}
  & Parentheses & No parentheses \\
  \addlinespace
  Marks & $\pospl^1\reqpl^{\mathbb{N}}$ or $\pospl^0\reqpl^{\mathbb{Z}_{+}}$ & $\pospl^1\reqpl^{\mathbb{N}}$ \\
  No marks & $\pospl^1\reqpl^{\mathbb{N}}$ or $\pospl^0\reqpl^{\mathbb{Z}_{+}}$ & $\pospl^1\reqpl^0$ \\
\end{tabular}
\end{center}

\noindent The absence of both non-grouping parentheses and marks mean that every occurrence of a pair of parentheses must permit zero or more pairs. Adding marks adds required pairs for some trees, turning each occurrence to one of $k$ or more pairs, for some non-negative integer $k$ determined by the tree. The left column is more complicated since non-grouping parentheses can introduce occurrences that require an exact number of parentheses, as well as all the complications of marks.

The remainder of this section will be devoted to first showing that $\neg \exists t_1, t_2 \in L(T_D).\ \words(t_1) \subseteq \words(t_2)$ implies that $D$ is resolvably ambiguous if $D$ is in the lower-right quadrant, i.e., if it has no marks and no non-grouping parentheses (in Section~\ref{sec:static-proofs}), and then describing the automaton that forms the basis for our algorithm (Section~\ref{sec:lattice-vpl}). Section~\ref{sec:static-proofs} additionally shows that adding marks to a resolvably ambiguous language in the lower-right quadrant preserves resolvability, which gives us a conservative approach for languages in the top-right quadrant. Intuitively, if a language has marks, but they are not required for resolvability, then we can still detect that the language is resolvable. For example, most expression languages fall in this category (we can write any expression unambiguously by adding parentheses everywhere, even without any precedence or associativity), while lists in OCaml do not (no amount of parentheses will make \verb|'[1;2]'| look like a list of two elements without a mark).

\subsection{The Road to Correctness} \label{sec:static-proofs} 

This section shows that $\neg \exists t_1, t_2 \in L(T_D).\ \words(t_1) \subseteq \words(t_2)$ implies that $D$ is resolvably ambiguous if $D$ has no non-grouping parentheses and no marks. These requirements can be written more formally as follows:

\begin{enumerate}
  \item $\T \cap \{\verb|'('|, \verb|')'|\} = \emptyset$.
  \item Every non-terminal $N_m$ on the right hand side of a productin in $D$ has $m = \emptyset$.
\end{enumerate}

\noindent Together, these imply that the canonical linear encoding of a tree has no ``$\reqp{}$'' (required parentheses), only ``$\posp{}$'' (optional parentheses).

Consider an arbitrary $t \in L(T_D)$ with canonical linear encoding $c$. Now replace each pair ``$\posp{}$'' with exactly one pair of parentheses, calling the result $w_\top$. We have $w_\top \in \words(t)$ by construction. Now consider the two possibilities of ambiguity for $w_\top$:

\begin{description}
  \item[$w_\top$ is ambiguous.] In this case, $\exists t'.\ t \neq t' \land t' \in \parse(w_\top)$. This means that we can take the canonical encoding of $t'$ (call it $c'$) and replace every ``$[]$'' with either one or zero parentheses and produce $w_\top$. But that also means that every word we construct from $c$ (by choosing some non-negative integer of parentheses for each ``$[]$'') can also be constructed from $c'$ (by choosing the same number for corresponding parentheses, and zero for the others), i.e., $\words(t) \subseteq \words(t')$.
  \item[$w_\top$ is unambiguous.] In this case $\words(t)$ has at least one unambiguous word ($w_\top$) which thus cannot be part of $\words(t')$ for any other $t' \in L(T_D)$, thus $\neg \exists t'.\ t \neq t' \land \words(t) \subseteq \words(t')$.
\end{description}

\noindent We thus arrive at the first of two central theorems for our approach:

\begin{theorem}
  Given a language with productions $D$ and terminals $\T$, where $\T \cap \{\verb|'('|, \verb|')'|\} = \emptyset$ and all non-terminals $N_m$ appearing in a right-hand side of a production in $D$ having $m = \emptyset$, the following holds: $\rho_P(D) <-> \neg \exists t_1, t_2 \in L(T_D).\ t_1 \neq t_2 \land \words(t_1) \subseteq \words(t_2)$.
\end{theorem}

Adding marks to such a language preserves resolvability, since:

\begin{itemize}
  \item Marks can introduce at most one required pair of parentheses per node.
  \item Adding marks to a language can only ever shrink $\words(t)$ for any particular $t \in L(T_D)$.
  \item Given a tree $t$, the word produced by putting one pair of parentheses around each node in $t$ (thus potentially adding double parentheses), call it $w$, is unambiguous in the language without marks. It cannot be excluded from $\words(t)$ by adding marks (since it already has a pair of parentheses around each node), and it also cannot be ambiguous, since that would require $\words(t')$ for some other tree $t'$ to have grown. Thus, $\parse(w) = \{t\}$ even when we add marks to $D$.
\end{itemize}

\begin{theorem}
  Adding marks to a resolvably ambiguous language $D$ preserves resolvability, if $D$ had no parentheses and no marks.
\end{theorem}

\subsection{The Algorithm} \label{sec:lattice-vpl}

This section describes the construction of a VPDA that we examine to determine if a language has a pair of trees where the language of one entirely contains the other. We begin by outlining the approach and then describe it in more detail.

The automaton we examine recognizes canonical linear encodings and has a one-to-one correspondence between runs and trees. Two distinct runs that recognize the same word thus corresponds to two distinct trees that have the \emph{same} word language. This is essentially the question of whether a VPDA is ambiguous, which we can detect by constructing a product automaton and trimming it.

To detect subsumption we make the following observation: given two trees $t_1$ and $t_2$ with canonical encodings $c_1$ and $c_2$, respectively, we have $\words(t_1) \subset \words(t_2)$ iff we can add arbitrary well-nested ``$\posp{}$'' pairs to $c_1$, producing a $c_1'$ such that $c_1' = c_2$ (Note that we do not need to consider ``$\reqp{}$'' since no such pairs appear in this language class). We thus make a product automaton of two slightly different automata, where the second may add arbitrary ``$\posp{}$'', but is otherwise identical.

\subsubsection{Illustration by Example}

We begin by constructing a VPDA that recognizes \emph{a} linear encoding for each tree, then modify it to only recognize the \emph{canonical} linear encoding. To illustrate the approach we consider the following language:

\begin{center}
  \begin{tabular}{@{}r@{$\;->\;$}c@{\,:\;}l}
    \toprule
    $S$ & $a$ & $E(\verb|';'| + \epsilon)$ \\
    $E$ & $b$ & $E\verb|'s'|$ \\
    $E$ & $c$ & \verb|'z'| \\
    \bottomrule
  \end{tabular}
\end{center}

\noindent This language has no non-grouping parentheses and no marks and is thus in the language class we consider. We begin by constructing a DFA per production (using standard methods, since the right hand side is a regular expression).

\begin{center}
\begin{tikzpicture}[->,>=Latex,shorten >=1pt,auto,node distance=2cm,
    scale = .9,transform shape]
  \node[state,initial] (1) {$1$};
  \node[state,accepting] (2) [right of=1] {$2$};
  \node[state,accepting] (3) [right of=2] {$3$};

  \node[state,initial] (4) [below of=1,yshift=7mm] {$4$};
  \node[state] (5) [right of=4] {$5$};
  \node[state,accepting] (6) [right of=5] {$6$};

  \node[state,initial] (7) [below of=4,yshift=7mm] {$7$};
  \node[state,accepting] (8) [right of=7] {$8$};

  \path (1) edge node{$E$} (2)
  (2) edge node{$';'$} (3);

  \path (4) edge node{$E$} (5)
  (5) edge node{$'s'$} (6);

  \path (7) edge node{$'z'$} (8);
\end{tikzpicture}
\end{center}

\noindent We then combine these separate DFAs to produce a single VPDA in three steps:

\begin{enumerate}
\item Introduce two new distinct states (call them $s$ and $f$) and a transition $s \xrightarrow{S} f$, where $S$ is the starting symbol of the language being examined.
\item Replace each transition $p \xrightarrow{N} q$ with:
  \begin{itemize}
  \item A transition $p \xrightarrow{'\pospl', +(p,q)} q_0$ for every initial state $q_0$ in a DFA corresponding to a production with lefthand side $N$.
  \item A transition $q_f \xrightarrow{'\pospr', -(p,q)} q$ for every accepting state $q_f$ in a DFA corresponding to a production with lefthand side $N$.
  \end{itemize}
\item Make $s$ the only initial state, and $f$ the only accepting state.
\end{enumerate}

\noindent To reduce clutter we omit $'\pospl'$ and $'\pospr'$, i.e., we abbreviate $p \xrightarrow{'\pospl', +(p',q')} q$ as $p \xrightarrow{+(p', q')} q$ and $p \xrightarrow{'\pospr', -(p',q')} q$ as $p \xrightarrow{-(p', q')} q$.

\begin{center}
\begin{tikzpicture}[->,>=Latex,shorten >=1pt,auto,node distance=2cm,
    scale = .9,transform shape]
  \node[state] (1) {$1$};
  \node[state] (2) [right of=1] {$2$};
  \node[state] (3) [right of=2] {$3$};

  \node[state] (4) [below of=1] {$4$};
  \node[state] (5) [right of=4] {$5$};
  \node[state] (6) [right of=5] {$6$};

  \node[state] (7) [below of=4] {$7$};
  \node[state] (8) [right of=7] {$8$};

  \node[state,initial] (s) [left of=1,xshift=-0.9cm] {$s$};
  \node[state,accepting] (f) [right of=3,xshift=0.9cm] {$f$};

  \path (2) edge node{$';'$} (3);

  \path (5) edge node{$'s'$} (6);

  \path (7) edge node{$'z'$} (8);

  \path (s) edge node{$+(s, f)$} (1);
  \draw (2.north) -- ([yshift=4mm]2.north) -| (f.north);
  \draw (3.north) -- ([yshift=4mm]3.north) -| (f.north);
  \node at ([yshift=5mm]$(3)!0.5!(f)$) {$-(s,f)$};

  \draw (1.east) -| ([xshift=4mm]4.east) |- ([xshift=4mm,yshift=10mm]7.east) -- (7);
  \draw ([xshift=4mm]4.east) -- (4);
  \draw (6.north) |- ([xshift=-8mm]$(2)!0.5!(5)$) |- (2.west);
  \draw (8) -- ([xshift=-3.5mm,yshift=10mm]8.west) |- (2.west);
  \node at ([xshift=2mm,yshift=2mm]$(1)!0.5!(4)$) {$+(1,2)$};
  \node at ([xshift=-2mm,yshift=2mm]$(2)!0.5!(5)$) {$-(1,2)$};

  \draw (4.south) -- (7.north);
  \path (4) edge[out=270,in=225,looseness=5] node[xshift=4mm]{$+(4,5)$} (4);
  \draw (6.south) -- ([yshift=-5mm]6.south) -| (5.south);
  \path (8) edge[right] node[yshift=-3mm]{$-(4,5)$} (5);
\end{tikzpicture}
\end{center}

\noindent This automaton recognizes a linear encoding by simply putting a pair $'\posp{}'$ around each production. That encoding will be canonical in many cases (e.g. ``$\posp{\posp{\posp{z}s};}$''), but not all; ``$\posp{\posp{z}}$'', which is recognized for the tree $a(c(z))$, is not canonical, it should be ``$\posp{z}$''. The extra $'\posp{}'$ pair stems from $S -> a: E(\verb|';'| + \epsilon)$, its righthand side matches $E$ (i.e., a word consisting of a single non-terminal).

We solve this by first recording which productions have righthand sides that match single non-terminal words (and which non-terminal) and then compute DFAs that do \emph{not} match such words. For our running example we note that the production $a$ matches the word $E$ (we record this as a tuple $(S, a, E)$) and replace its automaton with one that does not match $E$. In this case it is sufficient to mark state $2$ as non-accepting.

The recorded tuple $(S, a, E)$ can be read as ``starting from non-terminal $S$ we can match the word $E$ by choosing the production with label $a$.'' If we consider the middle element as a sequence of labels we can build a relation by adding $(N, \epsilon, N)$ for all non-terminals $N$ and computing the closure of $(A, w, B) \cdot (B, w', C) = (A, w \cdot w', C)$. Call the resulting set $T$.

Finally, we change step 2 in our construction to the following:

\begin{enumerate}
  \setcounter{enumi}{1}
\item Replace each transition $p \xrightarrow{N} q$ with:
  \begin{itemize}
  \item A transition $p \xrightarrow{'\pospl', +(p,q,w)} q_0$, for every initial state $q_0$ in a DFA corresponding to a production with lefthand side $N'$, for every $(N, w, N') \in T$.
  \item A transition $q_f \xrightarrow{'\pospr', -(p,q,w)} q$, for every accepting state $q_f$ in a DFA corresponding to a production with lefthand side $N'$, for every $(N, w, N') \in T$.
  \end{itemize}
\end{enumerate}

\begin{center}
\begin{tikzpicture}[->,>=Latex,shorten >=1pt,auto,node distance=2cm,
    scale = .9,transform shape]
  \node[state] (1) {$1$};
  \node[state] (2) [right of=1] {$2$};
  \node[state] (3) [right of=2] {$3$};

  \node[state] (4) [below of=1] {$4$};
  \node[state] (5) [right of=4] {$5$};
  \node[state] (6) [right of=5] {$6$};

  \node[state] (7) [below of=4] {$7$};
  \node[state] (8) [right of=7] {$8$};

  \node[state,initial] (s) [left of=1,xshift=-0.9cm] {$s$};
  \node[state,accepting] (f) [right of=3,xshift=0.9cm] {$f$};

  \path (2) edge node{$';'$} (3);

  \path (5) edge node{$'s'$} (6);

  \path (7) edge node{$'z'$} (8);

  \path (s) edge node{$+(s, f, \epsilon)$} (1)
  (3) edge node{$-(s, f, \epsilon)$} (f);

  \draw (1.east) -| ([xshift=4mm]4.east) |- ([xshift=4mm,yshift=10mm]7.east) -- (7);
  \draw ([xshift=4mm]4.east) -- (4);
  \draw (6.north) |- ([xshift=-8mm]$(2)!0.5!(5)$) |- (2.west);
  \draw (8) -- ([xshift=-3.5mm,yshift=10mm]8.west) |- (2.west);
  \node at ([xshift=1mm,yshift=2mm]$(1)!0.5!(4)$) {$+(1,2,\epsilon)$};
  \node at ([xshift=-1mm,yshift=2mm]$(2)!0.5!(5)$) {$-(1,2,\epsilon)$};

  \draw (4.south) -- (7.north);
  \path (4) edge[out=270,in=225,looseness=5] node[xshift=4mm]{$+(4,5,\epsilon)$} (4);
  \draw (6.south) -- ([yshift=-5mm]6.south) -| (5.south);
  \path (8) edge[right] node[yshift=-3mm]{$-(4,5,\epsilon)$} (5);

  \draw (s.south) |- (7.west);
  \draw (s.south) |- (4.west);
  \node at ([xshift=8mm,yshift=-10mm]s) {$+(s, f, a)$};
  \draw (8.east) -| (f.south);
  \draw (6.east) -| (f.south);
  \node at ([xshift=-8mm,yshift=-10mm]f) {$-(s, f, a)$};
\end{tikzpicture}
\end{center}

\noindent We call this automaton $A_{\posp{}}$. We also construct a modified automaton that recognizes linear encodings of larger word languages. We call this automaton $A'_{\posp{}}$ and obtain it by adding transitions $p \xrightarrow{'\pospl', +\gamma} p$ and $p \xrightarrow{'\pospr', -\gamma} p$ for every state $p$, where $\gamma$ is a new distinct stack symbol. Note that all these new transition push and pop the \emph{same} stack symbol, meaning that a new pair ``$\posp{}$'' can be started in any state and then closed in any state (not just the same state) as long as intermediate transitions leave the stack unchanged, i.e., only introduce zero or more well-balanced pairs.

\subsubsection{Formalization}

Given a language definition $D$, with terminals $\T$, non-terminals $\NT$, labels $\Labels$, and starting non-terminal $S$. We use $\labelnt(l)$ and $\labelregex(l)$ as the non-terminal and regex, respectively, of the production labelled $l$. We define the set $T : \NT \times \Labels^{*} \times \NT$ of transitions between productions inductively through:

\begin{itemize}
\item $(N, \epsilon, N) \in T$ for all $N \in \NT$.
\item $(N, l, N') \in T$ for all $N' \in \NT \cap L(\labelregex(l))$ where $\labelnt(l) = N$.
\item $(N_1, w_1 \cdot w_2, N_3) \in T$ for all $(N_1, w_1, N_2), (N_2, w_2, N_3) \in T$.
\end{itemize}

\noindent $T$ is finite if $D$ has no cycles $N =>^{+} N$. Such a cycle is easy to detect (depth-first search), and means that any tree that contains a production from the non-terminal $N$ only ever produces (infinitely) ambiguous words, i.e., all words for such a tree are unresolvably ambiguous.

We define the translation from production to DFA through a function $\dfa$ from a label to a DFA with the alphabet $\T \cup \Labels$:

\begin{description}
\item[$\dfa(l)$:] Compute the regular language $L' := L(\labelregex(l)) \setminus \NT$, then construct a DFA for this language.
\end{description}

\noindent We now define the VPDA $A_{\posp{}}$. Given $\dfa(l) = (Q_l, \T \cup \NT, \delta_l, s_l, F_l)$ for all $l \in Labels$, $A_{\posp{}} = (Q, \T', \Gamma, \delta, s, \{f\})$ where:

\begin{itemize}
\item $s$ and $f$ are two new distinct states.
\item $Q = \{s, f\} \cup \bigsqcup_l Q_l$.
\item $\T' = \T \cup \{'\pospl', '\pospr'\}$.
\item $\Gamma = Q \times Q \times \Labels^{*}$.
\item $\delta$ is defined by the following equations (unspecified cases produce $\emptyset$):
  $$
  \begin{array}{l}
    \delta(p, a, \lambda) = \{(\delta_l(p, a), \lambda)\} \qquad \text{ where } p \in Q_l \text{ and } a \in \T \\

    \delta(s, \text{``}\pospl\text{''}, \lambda) =
    \{ (s_{l}, (s, f, w)) \mid
    (S, w, N) \in T, N = \labelnt(l) \} \\

    \delta(p, \text{``}\pospl\text{''}, \lambda) = \\
    \qquad\{ (s_{l_3}, (p, q, w)) \mid
    p \in Q_{l_1}, \delta_{l_1}(p, N) = q, (N, w, N') \in T, N' = \labelnt(l_3) \} \\

    \delta(p, \text{``}\pospr\text{''}, (s, f, w)) =
    \{ (f, \lambda) \mid
    (S, w, N) \in T, N = \labelnt(l), p \in F_l \} \\

    \delta(q', \text{``}\pospr\text{''}, (p, q, w)) = \\
    \qquad \{ (q, \lambda) \mid
    q' \in F_{l_1}, \labelnt(l_1) = N', (N, w, N') \in T, p \in Q_{l_2}, \delta_{l_2}(p, N) = q \} \\

  \end{array}
  $$
\end{itemize}

\noindent We also construct a modified VPDA $A'_{\posp{}} = (Q, \T', \Gamma \cup \{\gamma\}, \delta', s, \{f\})$ where:

\begin{itemize}
\item $\gamma$ is a new distinct stack symbol.
\item $\delta'(p, a, g) = \delta(p, a, g) \cup \delta''(p, a, g)$, where $\delta''$ is given by:
  $$
  \begin{array}{r@{\;=\;}l@{}}
    \delta''(p, \text{``}\pospl\text{''}, \lambda)
    & \{ (p, \gamma) \} \\

    \delta''(p, \text{``}\pospr\text{''}, \gamma)
    & \{ (p, \lambda) \} \\
  \end{array}
  $$
\end{itemize}

\noindent To find two distinct successful runs, one in each automaton, that recognize the same word we construct the product automaton $A_{\posp{}} \times A'_{\posp{}} = (Q \times Q, \T', \Gamma \times (\Gamma \cup \{\gamma\}), \delta_{\times}, (s, s), \{(f, f)\})$, where $\delta_{\times}$ is described in \cite{alurVisiblyPushdownLanguages2004} (using the partitions $\T_c = \{ \pospl \}$, $\T_i = \T$, and $\T_r = \{\pospr\}$). If the product automaton has a successful run that passes through at least one configuration $((p, p'), (g, g')\cdot w)$ such that $p \neq p' \lor g \neq g'$ (distinct states, or distinct stack symbols, respectively), then that run corresponds to two distinct successful runs in $A_{\posp{}}$ and $A'_{\posp{}}$. We check for the existence of such a run by trimming the product automaton (as described in \cite{caralpTrimmingVisiblyPushdown2015}) and looking for a transition that pushes $(g, g')$ where $g \neq g'$, or transitions to $(q, q')$ where $q \neq q'$.

Performing this procedure on a language definition $D$ without non-grouping parentheses and marks gives a sound and complete decision procedure for determining if $D$ is resolvably ambiguous. If we instead take a language definition $D$ with marks but no non-grouping parentheses, then remove the marks and perform the procedure, then the absence of two such runs implies that $D$ is resolvably ambiguous, while their presence only implies the possibility of an unresolvable ambiguity; the marks might remove it.

\begin{lemma}
  There is a bijection between successful runs in $A_{\posp{}}$ and trees $t \in L(T_D)$.
\end{lemma}

\begin{proof}[sketch]
  By defining two functions $\treetorun$ and $\runtotree$, the former from a tree to a run, the latter from run to tree, then showing these two functions to be inverses of each other.
\end{proof}

We denote the word recognized by a run $R$ by $\runtoword(R)$.

\begin{lemma}
  $\forall t \in L(T_D).\ \runtoword(\treetorun(t))$ is the canonical linear encoding of $\words(t)$.
\end{lemma}

\begin{lemma}
  For all $t \in L(T_D)$ and every linear encoding $c$ such that $\words(t) \subseteq L(c)$, $c \in L(A'_{\posp{}})$.
\end{lemma}

\begin{lemma}
  We can define a function $\runprimetorun$ from runs in $A'_{\posp{}}$ to runs in $A_{\posp{}}$ such that $\forall R.\ L(\runtoword(\runprimetorun(R))) \subseteq L(\runtoword(R))$.
\end{lemma}

\begin{proof}[sketch]
  By removing every transition that pushes or pops $\gamma$.
\end{proof}

\begin{theorem}
  The existence of two distinct runs $R$ and $R'$ such that $R$ is a successful run in $A_{\posp{}}$ and $R'$ in $A'_{\posp{}}$, and $\runtoword(R) = \runtoword(R')$, implies the existence of two distinct trees such that the word language of one is entirely contained in the other.
\end{theorem}

\begin{proof}[sketch]
  Assume $\runprimetorun(R') = R$. $\runprimetorun$ either leaves its argument unchanged, or changes the recognized word. If $R'$ and $R$ recognize the same word, and $\runprimetorun(R') = R$, then $R' = R$, but that is a contradiction.

  Since $\runprimetorun(R') \neq R$ we have two distinct trees $t_1 = \runtotree(\runprimetorun(R'))$ and $t_2 = \runtotree(R)$. Since $L(\runtoword(\runprimetorun(R'))) \subseteq L(\runtoword(R'))$ and $L(\runtoword(R')) = L(\runtoword(R))$ we have that $L(t_1) \subseteq L(t_2)$.
\end{proof}

We thus have a way to detect the presence or absence of a pair of trees $t_1, t_2 \in L(T_D)$ such that $\words(t_1) \subseteq \words(t_2)$, which completely determines the resolvability of $D$ when $D$ has no non-grouping parentheses or marks.

\section{Case Studies} \label{sec:evaluation}

We have implemented a tool and a small DSL for syntax definition
(Section~\ref{sec:evaluation-syncon}) that generates a language
definition in the style of
Section~\ref{sec:parse-time-disambiguation}, which the tool then
uses to construct a parser and do dynamic resolvability analysis.
In this section, we showcase two possible use-cases of dynamic
resolvability analysis using this tool: composing separately
defined DSLs (Section~\ref{sec:evaluation-orc}), and finding
ambiguities in grammars through testing
(Section~\ref{sec:evaluation-ocaml}).

\subsection{A DSL for Syntax Definition}
\label{sec:evaluation-syncon}

We write our syntax definitions in a DSL building on
\emph{syncons}, introduced by Palmkvist and
Broman~\cite{palmkvistCreatingDomainSpecificLanguages2019},
wherein each production is defined separately from each other (one
\syncon{syncon} each).
The tool is implemented in Haskell and uses a custom implementation of
the Earley parsing algorithm \cite{earleyEfficientContextfreeParsing1970}
and a different version of the dynamic analysis that is faster in the
common case, but not guaranteed to terminate on an unresolvable
ambiguity. We bound the execution time with a timeout, and report
``unresolvable'' if this timeout is triggered.

The main difference between the DSL and the formalism in
Section~\ref{sec:parse-time-disambiguation} is that marks are
introduced as separate \syncon{forbid} declarations instead of
being inlined in productions. This allows supporting
convenience constructs, e.g., specifying precedence between
previously defined operators instead of manually inserting marks,
but is also important for composability as ambiguities can be
statically resolved without changing the original definitions.
Other convenience constructs include prefix, infix, and
postfix operators with a given associativity.
The running example used in
Section~\ref{sec:parse-time-disambiguation}
(Fig.~\ref{fig:running-example-definition} on
page~\pageref{fig:running-example-definition}) can be defined as
follows:

\begin{tabular}{ll}
\small
\begin{lstlisting}[language=syncon,boxpos=t]
type Exp
grouping "(" Exp ")"
precedence {
  mul; // higher in list means
  add; // higher precedence
}
\end{lstlisting}
&
\small
\begin{lstlisting}[language=syncon,boxpos=t]
token Integer = "[0-9]+"
syncon literal: Exp = n:Integer
syncon list: Exp =
  "[" (head:Exp (";" tail:Exp)*)? "]"
infix add: Exp = "+"
infix mul: Exp = "*"
\end{lstlisting}
\end{tabular}\smallskip

\noindent Here, we define a syntax type \syncon{Exp} for
expressions, and declare that parentheses can be used to group
expressions. We define a lexical token for integers, and a
\syncon{syncon} for integer literals using this token. A list is
defined as a bracketed sequence of zero or more expressions separated by
semi-colons. Finally, addition and
multiplication are defined as infix syncons with the expected
precedence rules. Note that we could also have replaced
the precedence list with explicit \syncon{forbid} declarations,
``\syncon{forbid mul.left = add}'' and ``\syncon{forbid mul.right = add}''
(cf. the mark $\{a\}$ in the production of $m$ in
Fig.~\ref{fig:running-example-definition}).



\subsection{Composing Language Definitions} \label{sec:evaluation-orc}

In this section, we consider the use case of defining a language
and composing it with another previously defined language, showing
how we can deal with the resulting ambiguities that arise from the
composition.
The language being defined is a subset of
Orc~\cite{kitchinOrc2009}, a functional programming language which
includes a number of special-purpose combinators for coordinating
concurrent workflows. On their own, these combinators act as a DSL
for concurrency.

In Orc, every expression can ``publish'' zero or more values. Orc
defines four combinators for orchestrating these published values:
the parallel (\ocaml{|}), sequential (\ocaml{>}$x$\ocaml{>}),
pruning (\ocaml{<}$x$\ocaml{<}), and ``otherwise'' (\ocaml{;})
combinators.
The expression $e_1$ \ocaml{|} $e_2$ runs $e_1$ and
$e_2$ in parallel, publishing any value published by either of
them.
The expression $e_1$ \ocaml{>}$x$\ocaml{>} $e_2$ executes
$[x\mapsto v]e_2$ for \emph{each} value $v$ published by $e_1$,
building a concurrent pipeline.
The expression $e_1$ \ocaml{<}$x$\ocaml{<} $e_2$ executes
$[x\mapsto v]e_1$ for \emph{first} value $v$ published by $e_2$
(discarding any remaining values published by $e_2$).
Finally, $e_1$\ocaml{;}$e_2$ runs $e_2$ only if running $e_1$ does
not publish any values.

The syncon definition of the Orc combinators is very simple (the
precedence rules follow the original
definition~\cite{kitchinOrc2009}):

\begin{tabular}{ll}
\small
\begin{lstlisting}[language=syncon,boxpos=t]
type Exp
grouping "(" Exp ")"
precedence {
  seq; par; prune; otherwise;
}
\end{lstlisting}
&
\small
\begin{lstlisting}[language=syncon,boxpos=t]
token Ident = "[[:lower:]][[:word:]]*"
infix par:Exp = "|"
infix seq:Exp = ">" x:Ident ">"
infix prune:Exp = "<" x:Ident "<"
infix otherwise:Exp = ";"
\end{lstlisting}
\end{tabular}\smallskip

\noindent
Na\"{i}vely composing this definition with another, separately
defined language is likely to introduce ambiguities. For example,
if that language also contains infix operators, the precedence
between these and the Orc combinators will be undefined. With
support for resolvable ambiguity, however, we can allow this
ambiguity and let programmers use parentheses for disambiguation.
After composing the above definition with a simple language
supporting addition, variables, and function calls (full
definition omitted for brevity), parsing the expression
``\ocaml{1 + 2 >x> f(x)}''
results in the following error message from our tool:

{\small
\begin{lstlisting}[language={[objective]caml}]
Ambiguity error with 2 alternatives:
  ( 1 + 2 ) > x > f ( x )
  1 + ( 2 > x > f ( x ) )
\end{lstlisting}
}

\noindent
Because there is no precedence specified between \ocaml{+} and
\ocaml{>x>}, disambiguation is required to specify the order of
operations.
In this case, it is likely that the preferred semantics are to have
all operators in the base language bind tighter than the Orc
combinators, and these precedence rules can be added after the
composition without changing the original definitions. Importantly
though, we are not \emph{required} to resolve this ambiguity at
time of composition.

Because of how we defined the syntax of the combinators, and since
the tool is currently whitespace insensitive, the
multi-character combinators (sequencing and pruning) are parsed as
three separate lexical tokens (this is visible in how whitespace
is inserted in the error message above). This means we can also
run into \emph{unresolvable} ambiguities, for example if our base
language includes comparison of numbers with \ocaml{<} and
\ocaml{>}. With such a base language (assuming no associativity
for \ocaml{>}), parsing the expression
``\ocaml{42 >x> f(x)}''
will result in the following error message:

{\small
\begin{lstlisting}[language={[objective]caml}]
Unresolvable ambiguity error with 2 alternatives.
\end{lstlisting}

\begin{minipage}{.3\textwidth}
\begin{lstlisting}[language={[objective]caml}]
Resolvable alternatives:
  ( 42 > x ) > f ( x )
  42 > ( x > f ( x ) )
\end{lstlisting}
\end{minipage}
\hfill
\begin{minipage}{.5\textwidth}
\begin{lstlisting}[language={[objective]caml}]
Unresolvable alternatives:
  seq
   - int        gt.orc:1:1-3
   - call       gt.orc:1:8-12
\end{lstlisting}
\end{minipage}
}

\noindent
By adding parentheses, a programmer can disambiguate the
expression as two comparisons. However, there is no way for a
programmer to specify that what they want is a sequential
combinator with left and right children being an integer and a
function call, respectively. In this case we have at least three
choices to make as language designers:

\begin{itemize}
\item We could decide to forbid the case where two comparisons are
  right next to each other, e.g., \syncon{forbid gt.left = gt} and
  \syncon{forbid gt.right = gt} (where \syncon{gt} is the syncon
  for the greater-than operator).
\item We could change the definition of the parallel combinators
  and define separate lexical tokens for the combinators, e.g,
  \syncon{token Seq = ">[[:lower:]]}\\\syncon{[[:word:]]*>"}. This
  would add just enough whitespace sensitivity to allow separating
  the different cases.
\item We could change the syntax of the combinators to avoid
  clashes.
\end{itemize}

\noindent
The first alternative is somewhat ad-hoc, but can be done after composition
and would allow us to reuse both language definitions without
modifications. In the latter two alternatives, we lose reuse of
the definitions of Orc combinators, but place no additional
restrictions on the base language. The preferred resolution
strategy is likely to differ between different compositions and
different languages.
For example, another unresolvable ambiguity is going to show up if
the base language uses semi-colons, e.g. for sequencing
expressions (this would clash with the ``otherwise'' combinator).
In this case, there is no reasonable way to disambiguate
``\ocaml{e1; e2}'' without changing the syntax of one of the
operations.

The main takeaway from this case study is that resolvable
ambiguity makes composition of languages less restrictive than if
all ambiguity is completely banned. The approach is strictly more
general since dynamic resolvability analysis allows deferring
disambiguation to the programmer, while removing ambiguities from
the resulting grammar is also possible.

\subsection{Finding Ambiguities with Dynamic Resolvability Analysis} \label{sec:evaluation-ocaml}

In this section, we investigate how to apply dynamic resolvability
analysis to find ambiguities in the specification of a
real-world language.
Using syncons, we have specified a substantial subset of OCaml's
syntax as described in chapters 7 and 8 (base language and
extensions) of the OCaml reference manual~\cite{leroyOCamlSystemRelease2018}.
The syncon definition is just under 700 lines long, and can
currently parse roughly 75\% (1012 out of 1334
files) of the \verb|.ml| files present in the OCaml compiler
itself (we discuss limitations of our tool in the end of this
section).

By using our implementation of the OCaml syntax to parse
real-world OCaml code, we can dynamically identify ambiguities in
the grammar presented in the reference manual, instead of silently
resolving these in the implementation of an unambiguous parser.
For example, an OCaml expression can be a function application
$\mathit{expr}~\{\mathit{argument}\}^{+}$ or a constructor
application $\mathit{constr}~\mathit{expr}$. Since an expression
can also be a constructor, however, the language definition allows
``\ocaml{Foo 42}'' to be parsed both as a function application and
as a constructor application (only the latter is well-typed, but
both are \emph{syntactically} well-formed).
Similarly, the expression ``\ocaml{Foo.f}'' can be parsed as a
field access targeting either a constructor or a module, since
both must start with capital letters (again, only the latter is
well-typed).
Another example, that was already mentioned in
Section~\ref{sec:parse-time-disambiguation}, is the fact that
semi-colons are used both for sequencing and for separating
elements in lists, arrays and records. Thus, ``\ocaml{[1; 2]}''
can be parsed both as a list of two elements and a singleton list
equivalent to ``\ocaml{[(1; 2)]}''.

As this case study shows, implementing a parser that reports
ambiguities lets us identify ambiguities in a grammar through
testing. Even though detecting ambiguities in a grammar is
undecidable in general, dynamic resolvability analysis over a
large corpus of code lets us find (and fix) many cases of
ambiguity in practice.
In this case study, out of the 700 lines of syncon code, $\sim$200
lines are additions to conform to how the canonical compiler
behaves on cases that are under-specified in the
manual, including the examples listed above.


\paragraph{Current Limitations of Syncons}
Other than syntactic OCaml constructs that are not yet specified,
there are two sources of failure in the parts of the OCaml
compiler that we cannot parse.
First, our parsing tool does not support specifying ``longest
match'' on a production, which is required to handle the pattern
matching constructs correctly. Some cases can be handled using
\syncon{forbid} declarations, but it does not get us all the way.
Second, our system uses a different definition of precedence than
the OCaml language. Our translation from precedence to implicit
\syncon{forbid} declarations is shallow (it only considers direct
children), while the OCaml has deep precedence. For example,
addition binds stronger than \ocaml{let} (which syntactically
functions as a prefix operator in OCaml), and thus
``\ocaml{1 + let x = 1 in x + 2}'' should be parsed as
``\ocaml{1 + (let x = 1 in (x + 2))}''. Our tool, however, reports
the expression as ambiguous, additionally suggesting the
alternative interpretation
``\ocaml{(1 + (let x}\\\ocaml{= 1 in x)) + 2}''.
This interpretation is indeed valid if we only look at the direct
children of each operator, since ``\ocaml{(_ + _) + 2}'',
``\ocaml{1 + let _ in _}'', and ``\ocaml{let x = 1 in x}'' are all
individually correct with regards to precedence.



\section{Related Work}\label{sec:related-work}

Our related work falls in three categories: syntax definition formalisms (including subclasses of CFGs), language frameworks, and other approaches to ambiguity.

Afroozeh et al.~\cite{afroozehSafeSpecificationOperator2013}'s operator ambiguity removal patterns bear a striking resemblance to the marks presented in this paper. However, in special-casing (what in this paper would be) marks on left and right-recursions in productions they correctly cover the edge case involving deep precedence discussed in Section~\ref{sec:evaluation-ocaml}. This approach thus suggests an interesting direction for future work: extend the algorithms presented in this paper to cover it.

Danielsson and Norell~\cite{danielssonParsingMixfixOperators2011} give a method for specifying grammars for expressions containing mixfix operators. They allow non-transitive, non-total precedence, and ``feel that it is overly restrictive to require the grammar to be unambiguous.'' Similar to our approach, they do not reject ambiguous grammars, only ambiguous parses. They also introduce a concept of \emph{precedence correct} expressions; expressions where direct children must have higher precedence than parents. This is more restrictive than our approach, e.g., in a language where \verb|'+'| and \verb|'*'| have no defined relative precedence they reject \verb|'1 + 2 * 3'| as syntactically invalid, while we parse it as an ambiguous expression.

Parsing expression grammars \cite{fordParsingExpressionGrammars2004} sidestep the issue of ambiguity by not introducing it at all. However, this also loses the potential gains of leaving certain ambiguities. Additionally, since the ordering of productions matter, composition of languages must be ordered, and the interactions between composed languages becomes non-obvious, e.g., merely adding productions may remove previously recognized words from the language, depending on where they are added.

Most commonly used parser generators are based on unambiguous CFG subclasses, e.g., LL(k), LR(k), or LR(*). Others do not fit neatly in the Chomsky hierarchy, but still produce a single parse tree per parse, e.g., LL(*) \cite{parrLLFoundationANTLR2011} and ALL(*) \cite{parrAdaptiveLLParsing2014}. Yet others produce multiple parse trees or other forms of parse forests, e.g., GLR \cite{langDeterministicTechniquesEfficient1974}, GLL \cite{scottGLLParsing2010}, and Earley \cite{earleyEfficientContextfreeParsing1970}.

Silver \cite{vanwykSilverExtensibleAttribute2010}, a system for defining extensible languages using attribute grammars, and its associated parser Copper \cite{vanwykContextawareScanningParsing2007} have a ``Modular Well-Definedness Analysis'' \cite{kaminskiModularWellDefinednessAnalysis2013}, the syntactic component of which can be found in \cite{schwerdfegerVerifiableCompositionDeterministic2009}. This analysis guarantees that the composition of a base language and any number of extensions that have passed the analysis will compose to a grammar in LALR(1). This language class is more restrictive than both unambiguous and resolvably ambiguous languages, though somewhat comparable to the subclass our static analysis supports.

The detection of classical ambiguity in context-free grammars is undecidable in general \cite{cantorAmbiguityProblemBackus1962}, yet numerous heuristic approaches exist. Examples include linguistic characterizations and regular language approximations \cite{brabrandAnalyzingAmbiguityContextFree2007}, using SAT-solvers \cite{axelssonAnalyzingContextFreeGrammars2008}, and other conservative approaches \cite{schmitzConservativeAmbiguityDetection2007}, For an overview, and additional approaches, see the PhD thesis of Basten~\cite{bastenAmbiguityDetectionProgramming2011}.

Numerous language development frameworks and libraries support syntactic language composition \emph{without} any guarantees on the resulting language (e.g., \cite{heeringSyntaxDefinitionFormalism1989,lorenzenSoundTypedependentSyntactic2016,erdwegLayoutSensitiveGeneralizedParsing2013,katsSpoofaxLanguageWorkbench2010}). These systems tend to have some form of general parser, so that they can handle arbitrary context-free grammars, but mention no handling of ambiguities encountered by an end-user.


\section{Conclusion}\label{sec:conclusion}
In this paper, we introduce the concept of \emph{resolvable
  ambiguity}. A language grammar is resolvably ambiguous if all
ambiguities can be resolved by the end-user at parse time. This
approach departs from the common standpoint that grammars and
syntax definitions of languages must be unambiguous.
As part of the new concept, we formalize the fundamental
\emph{resolvable ambiguity problem}, divide it into static and
dynamic parts, and provide solutions for both variants for a
restricted class of languages. Through case studies, we show
practical applicability of the approach, both for building new
domain-specific languages and for reasoning about existing
general-purpose languages.

In future work, we will investigate weakening the restrictions
currently needed for the static and dynamic resolvability
algorithms to work. We also intend to develop our theory and tools
to support handling deep precedence ambiguities and ambiguities
based on ``longest match''. A long term goal is to be able to
suggest changes to the grammar of a language, based on ambiguities
found through dynamic analysis.

\bibliographystyle{splncs04}

\bibliography{All}

\pagebreak
\appendix

\section{Preliminaries}
\label{app:prel}

This appendix briefly describes some of the theoretical
foundations we build upon.

\subsection{Context-Free Grammars} \label{sec:preliminaries-cfgs}

A context-free grammar (CFG) $G$ is a 4-tuple $(\NT, \T, P, S)$ where $\NT$ is a set of non-terminals; $\T$ a set of terminals, disjoint from $\NT$; $P$ a finite subset of $\NT \times (\NT \cup \T)^{*}$, i.e., a set of productions; and $S \in \NT$ the starting non-terminal.

%
A word $w \in \T^{*}$ is recognized by $G$ if there is a sequence of steps starting with $S$ and ending with $w$, where each step replaces a single non-terminal using a production in $P$. Such a sequence is called a \emph{derivation}. The set of words recognized by $G$ is written $L(G)$


The standard definition of ambiguity, given a context-free grammar $G$, is expressed in terms of \emph{leftmost derivations}. A leftmost derivation is a derivation where the non-terminal being replaced is always the leftmost one.

\begin{definition}
A word $w \in L(G)$ is ambiguous if there are two distinct leftmost derivations of $w$.
\end{definition}

\subsection{Automata} \label{sec:preliminaries-automata}

A nondeterministic finite automaton (NFA) is a 5-tuple $(Q, \T, \delta, q_0, F)$ where $Q$ is a finite set of states; $\T$ a finite set of terminals; $\delta$ a transition function from $Q \times \T$ to finite subsets of $Q$; $q_0 \in Q$ an initial state; and $F \subseteq Q$ a set of final states. A successful run is a sequence of states $r_0, \ldots, r_n$ and a word $a_0\cdots a_n$ such that $r_0 = q_0$, $\forall i \in \{0, 1, \ldots, n-1\}.\ r_{i+1} \in \delta(r_i, a_i)$ and $r_n \in F$. We say that the automaton accepts the word $a_0a_1\cdots a_n$ iff there is such a successful run.

A deterministic finite automaton (DFA) has the same definition, except $\delta : Q \times \Sigma -> Q$, i.e., given a state and a symbol there is always a single state we can transition to. NFAs and DFAs have the same expressive power as regular expressions, i.e., for every regular expression there is an NFA and a DFA both reconizing the same language, and vice-versa. 

A pushdown automaton extends a finite automaton with a stack to/from which transitions can push/pop symbols. Formally, a nondeterministic pushdown automaton is a 6-tuple $(Q, \T, \Gamma, \delta, q_0, F)$ where $Q$ is a finite set of states; $\T$ a finite set of input symbols, i.e., an input alphabet; $\Gamma$ a finite set of stack symbols, i.e., a stack alphabet; $\delta$ a transition function from $Q \times (\T \cup \{\lambda\}) \times (\Gamma \cup \{\lambda\}$ to finite subsets of $Q \times (\Gamma \cup \{\lambda\}$; $q_0 \in Q$ the initial state; and $F \subseteq Q$ a set of final states. $\lambda$ essentially means ``ignore'', i.e., $\delta(q_1, \lambda, \lambda) = \{(q_2, \lambda)\}$ means: transition from state $q_1$ without consuming an input symbol (the first $\lambda$) and without examining or popping from the current stack (the second $\lambda$), to state $q_2$ without pushing a new symbol on the stack (the third $\lambda$).

A successful run is now a sequence of \emph{configurations}, elements of $Q \times \Gamma^{*}$, starting with $(q_0, \epsilon)$, ending with $(f, \gamma)$ for some $f \in F$ and $\gamma \in \Gamma^{*}$.

However, in this paper we only consider pushdown automata with relatively limited stack manipulation, and thus use some convenient shorthands:

\begin{itemize}
\item $p \xrightarrow{a} q$, a transition that recognizes the terminal $a$ and does not interact with the stack at all, i.e., $\delta(p, a, \lambda) \supseteq \{(q, \lambda)\}$.
\item $p \xrightarrow{a, +g} q$, a transition that recognizes the terminal $a$ and pushes the symbol $g$ on the stack, i.e., $\delta(p, a, \lambda) \supseteq \{(q, g)\}$.
\item $p \xrightarrow{a, -g} q$, a transition that recognizes the terminal $a$ and pops the symbol $g$ from the stack, i.e., $\delta(p, a, g) \supseteq \{(q, \lambda)\}$.
\end{itemize}

\subsection{Visibly Pushdown Languages} \label{sec:preliminaries-vpls}

A visibly pushdown language \cite{alurVisiblyPushdownLanguages2004} is a language that can be recognized by a visibly pushdown automaton. A visibly pushdown automaton (VPDA) is a pushdown automaton where the input alphabet $\T$ can be partitioned into three disjoint sets $\T_c$, $\T_i$, and $\T_r$, such that all transitions in the automaton have one of the following three forms:

\begin{itemize}
\item $p \xrightarrow{c, +s} q$, where $c \in \T_c$ and $s \in \Gamma$; or
\item $p \xrightarrow{i} q$, where $i \in \T_i$; or
\item $p \xrightarrow{r, -s} q$, where $r \in \T_r$ and $s \in \Gamma$,
\end{itemize}

\noindent i.e., the terminal recognized by a transition fully determines the change to the stack height.
The names of the partitions stem from their original use in program analysis, $c$ is for \emph{call}, $i$ for \emph{internal}, and $r$ for \emph{return}.
This partitioning gives us the following particularly relevant properties:

\begin{itemize}
\item Visibly pushdown languages with the same input partitions
  are closed under intersection, complement, and union
  \cite{alurVisiblyPushdownLanguages2004}. Intersection is given
  by a product automaton. Given a pair of VPDAs
  $(Q_1, \T, \delta_1, q_0, F_1)$ and
  $(Q_2, \T, \delta_2, q'_0, F_2)$ their product automaton has the
  form $(Q_1 \times Q_2, \T, \delta', (q_0, q'_0),$
  $ F_1 \times F_2)$ where, when $c \in \T_c$, $i \in \T_i$ and
  $r \in \T_r$:
  \[
  \begin{array}{l}
    \delta'((p_1, p_2), c, \lambda) =\\ \qquad\{((q_1, q_2), (g_1, g_2)) \mid (q_1, g_1) \in \delta_1(p_1, c, \lambda), (q_2, g_2) \in \delta_2(p_2, c, \lambda) \}\\
    \delta'((p_1, p_2), i, \lambda) =\\ \qquad \{((q_1, q_2), \lambda) \mid (q_1, \lambda) \in \delta_1(p_1, i, \lambda), (q_2, \lambda) \in \delta_2(p_2, i, \lambda) \}\\
    \delta'((p_1, p_2), r, (g_1, g_2)) =\\ \qquad \{((q_1, q_2), \lambda) \mid (q_1, \lambda) \in \delta_1(p_1, r, g_1), (q_2, \lambda) \in \delta_2(p_2, r, g_2) \}\\
  \end{array}
  \]

\item A VPDA can be trimmed \cite{caralpTrimmingVisiblyPushdown2015}, i.e., modified in such a way that all remaining states and transitions are part of at least one successful run; none are redundant. Furthermore, a successful run in the trimmed automaton corresponds to exactly one successful run in the original automaton, and vice-versa.
\end{itemize}

\subsection{Unranked Regular Tree Grammars} \label{sec:preliminaries-trees}

Trees generalize words by allowing each terminal to have multiple ordered successors, instead of just zero or one. Most literature considers \emph{ranked} tree languages, where each terminal has a fixed arity, i.e., the same terminal must always have the same number of successors. This is as opposed to \emph{unranked} tree languages, where the arity of a terminal is not fixed. The sequence of successors to a single terminal in an unranked tree tends to be described by a word language (referred to as a horizontal language in \cite{comonTreeAutomataTechniques2007}), often a regular language.

The results and properties presented in this paper are more naturally described through unranked trees, thus all references to trees here are to unranked trees, despite ranked being more common in the literature. We further distinguish terminals used solely as leaves from terminals that may be either nodes or leaves. Since we will use unranked trees to represent parse trees, the former will represent terminals from the parsed word, while the latter represent terminals introduced as internal nodes.

An unranked tree grammar $T$ is a tuple $(\NT, \T, X, P, S)$ where:

\begin{itemize}
\item $\NT$ is a set of (zero-arity) non-terminals.
\item $\T$ is a set of zero-arity terminals, used as leaves.
\item $X$ is a set of terminals without fixed arity, used as inner nodes or leaves.
\item $P$ is a set of productions, a finite subset of $\NT \times X \times \regex(\T \cup X)$. We will write a production $(N, x, r)$ as $N -> x(r)$.
\end{itemize}

\noindent A tree $t$ (containing only terminals from $\T$ and $X$) is recognized by $T$ if there is a sequence of steps starting with $S$ and ending with $t$, where each step either replaces a single non-terminal using a production in $P$, or replaces a regular expression $r$ with a sequence in $L(r)$. The set of trees recognized by $T$ is written $L(T)$\footnote{Again, to distinguish from regular expressions and context-free languages, all trees will be named $T$, possibly with a subscript.}. Finally, $yield : L(T) -> \T^{*}$ is the sequence of terminals $a \in \T$ obtained by a left-to-right\footnote{Preorder, postorder, or inorder does not matter since terminals in $\T$ only appear as leaves} traversal of a tree. Informally, it is the flattening of a tree after all internal nodes have been removed.

\end{document}